\PassOptionsToPackage{svgnames}{xcolor}
\documentclass[a4paper, 11pt]{article}
\usepackage[T1]{fontenc}

\usepackage[english]{babel}

\usepackage{todonotes}

\usepackage{latexsym}
\usepackage{amsmath}
\usepackage{amsthm}
\usepackage{amssymb}
\usepackage{mathtools}
\usepackage{xspace}
\usepackage{complexity}

\usepackage[ruled,linesnumbered]{algorithm2e}
\SetKwComment{Comment}{/* }{ */}

\usepackage{hyperref}
\usepackage[capitalize]{cleveref}
\usepackage{url}
\usepackage{comment}

\usepackage{tikz}
\usetikzlibrary{shapes,decorations}
\usetikzlibrary{graphs, graphs.standard}
\usetikzlibrary{positioning,petri}
\usetikzlibrary{scopes}
\usetikzlibrary{fit,backgrounds}
\usetikzlibrary{patterns,arrows.meta}
\tikzstyle{none}=[]
\tikzstyle{base}=[circle, fill=black!12, draw, inner sep=0pt, minimum width=8pt, minimum height=8pt, line width=0.5pt]
\tikzstyle{wideBase}=[circle, fill=black!12, draw, inner sep=0pt, minimum width=12pt, minimum height=12pt, line width=0.5pt]
\tikzstyle{dashStyle}=[-, dashed]
\tikzstyle{dotStyle}=[-,dotted]
\tikzstyle{baseEdge}=[-, draw=black, line width=1pt]
\tikzstyle{arrow}=[-{Latex[length=2mm]}, draw=black, line width=2pt]
\tikzstyle{slimArrow}=[->, draw=black, line width=1pt]
\tikzstyle{nodeLabel}=[shape=rectangle, fill=white, minimum width=8pt, minimum height=8pt, inner sep=0pt]
\pgfdeclarelayer{edgelayer}
\pgfdeclarelayer{nodelayer}
\pgfsetlayers{edgelayer,nodelayer,main}

\definecolor{softgreen}{HTML}{2cab27}
\definecolor{sanguine}{HTML}{eb5a21}
\definecolor{softyellow}{HTML}{e6e337}
\definecolor{coolpink}{RGB}{228, 127, 226}
\definecolor{coolgreen}{RGB}{127, 211, 125}
\definecolor{coolbrown}{RGB}{204, 174, 161}
\definecolor{coolorange}{RGB}{248, 185, 126}
\definecolor{coolred}{RGB}{255, 127, 125}
\definecolor{darkcoolred}{RGB}{255, 74, 71}
\definecolor{coolblue}{HTML}{00b4d8}
\definecolor{darkcoolblue}{HTML}{0065d8}

\usepackage{subcaption}

\title{A simple quadratic kernel\\ for Token Jumping on surfaces}
\author{Daniel W. Cranston%
\thanks{%
 Department of Computer Science, Virginia Commonwealth
 University, Richmond, VA, USA;
 \texttt{dcranston@vcu.edu}
 }
 \and
Moritz Mühlenthaler%
\thanks{%
 Laboratoire G-SCOP, Grenoble INP,
 Grenoble, France;
 \texttt{moritz.muhlenthaler@grenoble-inp.fr}
 }
 \and
Benjamin Peyrille%
\thanks{%
 Laboratoire G-SCOP, Grenoble INP, Universit\'{e} Grenoble-Alpes,
 Grenoble, France;
 \texttt{benjamin.peyrille@grenoble-inp.fr}
 }
}

\date{4 August 2024} 

\newtheorem{theorem}{Theorem}

\newtheorem{lemma}[theorem]{Lemma}

\newtheorem{corollary}[theorem]{Corollary}

\crefname{lemma}{Lemma}{Lemmas}
\crefname{theorem}{Theorem}{Theorems}
\crefname{corollary}{Corollary}{Corollaries}
\crefname{definition}{Definition}{Definitions}
\crefname{conjecture}{Conjecture}{Conjectures}
\crefname{figure}{Figure}{Figures}
\crefname{algorithm}{Algorithm}{Algorithms}
\crefname{remark}{Remark}{Remarks}
\crefname{claim}{Claim}{Claims}

\newcommand{\TJ}{\textsc{Token Jumping}\xspace}

\newcommand{\YES}{\textsc{Yes}\xspace}

\def\aftermath{\par\vspace{-\belowdisplayskip}\vspace{-\parskip}\vspace{-\baselineskip}}

\begin{document}

\maketitle

\begin{abstract}
    The problem \TJ asks whether, given a graph $G$ and two independent sets of \emph{tokens} $I$ and $J$ of $G$, we can transform $I$ into $J$ by changing the position of a single token in each step and having an independent set of tokens throughout. We show that there is a polynomial-time algorithm that, given an instance of \TJ, computes an equivalent instance of size $O(g^2 + gk + k^2)$, where $g$ is the genus of the input graph and $k$ is the size of the independent sets.
    Our algorithm is very simple and does not require an embedding of the input graph.
\end{abstract}

\section{Introduction}

We consider an independent set of a graph to be a set of \emph{tokens} placed on pairwise non-adjacent vertices. From an independent set we may obtain another one by letting a token ``jump'' to some vertex that has no token on it and no token in its neighborhood. We say that two independent sets are \emph{TJ-equivalent}, if one can be obtained from the other by a sequence of jump-operations. The problem \TJ asks whether two given independent sets of a graph are TJ-equivalent. \TJ can be thought of as a motion planning problem on a graph, where the tokens correspond to moving agents. It is a variant of the \emph{independent set reconfiguration} problem, which has attracted considerable attention in the last ten years, in particular in the context of parameterized complexity~\cite{ISRsurvey,Heuvel:13,Nishimura:18}.

\TJ is known to be \PSPACE-complete even on subcubic planar graphs of bounded bandwidth~\cite{Zanden:15}. Furthermore, the problem is \W[1]-hard when parameterized by the number $k$ of tokens and the number $\ell$ of jump-operations~\cite{Mouawad:17}. Notably, there are several positive results for sparse graphs: \TJ parameterized by $k$ admits a linear kernel on graphs of bounded degree~\cite{Bartier:21} and a polynomial kernel on graphs of bounded degeneracy~\cite{Sparse:18}. Furthermore, Ito, Kami\'nski, and Ono showed that \TJ is fixed-parameter tractable on planar graphs and, more generally, on $K_{3,t}$-free graphs~\cite{Ito:14}.
Finally, Bousquet, Mary, and Parreau gave a polynomial kernel for $K_{t,t}$-free graphs~\cite{bousquet2017token}, which implies a polynomial kernel for graphs embeddable on a fixed surface. However, when the problem is parameterized by the genus $g$ of the surface (rather than forbidding $K_{t,t}$, where $t$ is fixed) and the size $k$ of the independent sets, their kernel is not polynomial.
Our main result is the following theorem.

\begin{theorem}
    \label{thm:main-result-kernel}
    \TJ parametrized by the size $k$ of the independent sets and the genus $g$ of the input graph admits a kernel of size $O(g^2 + gk + k^2)$.
\end{theorem}

Our kernelization algorithm does not require any information about the genus of the input graph, which is \NP-hard to compute, unlike for instance the size $k$ of the independent sets.
The general algorithmic idea, which is already present in~\cite{Ito:14}, is the following. Consider an instance of \TJ given by a graph $G$ and two independent sets $I$ and $J$ of $G$, each of size $k$. For $Y \subseteq I \cup J$, let $C_Y$ be the set of vertices of $G$, whose set of neighbors is precisely $Y$. In order to obtain the kernel of~\cref{thm:main-result-kernel}, we partition the vertices not in $I$ and $J$ into three classes, depending on how many neighbors in $I \cup J$ they have:
\[
    \mathcal{C}_1 := \bigcup_{Y \subseteq I\cup J, |Y|\leq 1} C_Y  \hspace{30pt} \mathcal{C}_2 := \bigcup_{Y \subseteq I\cup J, |Y| = 2} C_Y \hspace{30pt} \mathcal{C}_3 := \bigcup_{Y \subseteq I\cup J, |Y| \geq 3} C_Y
\] 
We then show that $|\mathcal{C}_1| = O(\sqrt{g}\cdot k)$ or we can answer \textsc{Yes} (\cref{lem:c1-bound}), $|\mathcal{C}_3| = O(g^2 + gk + k)$ (\cref{lem:c3-bound}), and that we can replace $\mathcal{C}_2$ by a set of $O(g^2 + gk + k^2)$ vertices. The bounds on $\mathcal{C}_1$ and $\mathcal{C}_3$ are obtained by Heawood's and Euler's formulas, respectively. 
The algorithm we propose uses $\mathcal{C}_2$ as a buffer space for showing the jump-equivalence of $I$ and $J$. Our main contribution is a simple algorithm that reduces the size of $\mathcal{C}_2$.  Our analysis of this algorithm is elementary in the sense that it neither requires a drawing of the input graph nor makes use of any Ramsey-type result. 

Ito, Kami\'nski, and Ono use a bound on the Ramsey numbers to show that for planar graphs, if the subgraph induced by the neighbors of a pair $Y \in \mathcal{C}_2$ is larger than some $f(k)$, then it can be replaced by an independent set of size $k$ \cite{Ito:14}. Using this approach, they obtain a kernel of size $O(2^{6k})$ for planar graphs. For $K_{3,t}$-free graphs they do not give an explicit bound, but the reduced instance seems to be of size at least $2^{2k+(k+t+1)^{t+2}}$. Bousquet, Mary, and Parreau use a theorem of Kövári, Sós, and Turán to obtain a kernel for $K_{t, t}$-free graphs of size $O(f(t)\cdot {k^{t\cdot 3^t}})$~\cite{bousquet2017token}. 
We propose a simple polynomial-time algorithm that reduces the size of $\mathcal{C}_2$ to $O(g^2 + gk + k^2)$. 
This significantly improves on the results of~\cite{Ito:14} for planar graphs and~\cite{bousquet2017token} for bounded-genus graphs. Furthermore, since we do not require any arguments from extremal graph/set theory, our constants are reasonably small.

Fix $Y$ such that $\mathcal{C}_Y \in \mathcal{C}_2$. Since the input graph $G$ admits a crossing-free drawing on a surface of genus $g$, the subgraph of $G$ induced by $C_Y \cup Y$ can be partitioned into at most $4g$ homotopy classes (the endpoints of each curve are the vertices of $Y$). Each homotopy class corresponds to a planar subgraph of $G$ with a rather simple structure: any vertex outside of $C_Y$ that is adjacent to at least $3$ vertices of $C_Y$ cannot be adjacent to vertices of $C_Y$ that are not on some outer face. 
Using this, we can obtain in polynomial time a large enough linear forest of $C_Y$ which gives us an independent set $T_Y$ of size $2k+2$ such that no vertices outside of $C_Y$ are adjacent to more than two vertices of $T_Y$. We then replace all such $C_Y$ by $T_Y$.

A direct consequence of \cref{thm:main-result-kernel} is the following result, which improves on the polynomial kernel given in~\cite{bousquet2017token}.

\begin{corollary}
\label{cor:linear-kernel-bounded}
\TJ parameterized by the size $k$ of the independent sets admits a kernel of size $O(k^2)$ on graphs of bounded genus.
\end{corollary}

Using the same general approach and a very simple analysis, we also obtain a sub-quadratic kernel for $K_{2,3}$-free graphs, which notably include outerplanar graphs.

\section{Preliminaries}
\label{sec:prelim}

All graphs in this work are finite, undirected, and simple. Let $G = (V, E)$ be a graph. A set $I \subseteq V$ of vertices is an \emph{independent set} if the vertices in $I$ are pairwise non-adjacent. A proper vertex coloring of $G$ with $k$ colors is a partition of $V$ into $k$ independent sets. Two independent sets $I$ and $J$ of $G$ are \emph{TJ-equivalent} if there exists a sequence $I_1, I_2, \ldots, I_\ell$ of independent sets of $G$, such that $I = I_1$ and $J = I_\ell$ and for $1 \leq i < \ell$, we have $|I_i \setminus I_{i+1}| = |I_{i+1} \setminus I_i| =1$. Intuitively, at each step $i$, the independent set $I_{i+1}$ is obtained from $I_i$ by changing the position of a single token. 
Given a graph $G$ together with two independent sets $I$ and $J$ of $G$, each of size $k$, the problem \TJ asks whether $I$ and $J$ are TJ-equivalent.

A decision problem is \emph{fixed-parameter tractable (FPT)} for some parameter if there exists an algorithm that decides an instance of size $n$ with parameter $k \in \mathbb{N}$ in time $f(k)\cdot \poly(n)$, where $f : \mathbb{N} \to \mathbb{N}$ is some computable function. A \emph{kernelization algorithm} is a polynomial-time algorithm that takes as input an instance of size $n$ with parameter $k$ and outputs an equivalent instance (called \emph{kernel}) of size $g(k)$, where $g : \mathbb{N} \to \mathbb{N}$ is some computable function.

\section{A kernel for \TJ on surfaces}
\label{sec:kernel}

In this section we prove our main result, \cref{thm:main-result-kernel}. In the following, let $(G, I, J)$ be an instance of \TJ, where $G = (V, E)$ is a graph of genus $g$ and $I$ and $J$ are independent sets of $G$, each of size $k$. Let $X = I \cup J$ and notice that $|X| \leq 2k$.
The \emph{Heawood number} $H(g)$ is given by $H(g) = \left\lfloor (7 + \sqrt{1 + 48g})/{2} \right\rfloor$.  Heawood's proof from 1890~\cite{Heawood} and the four color theorem imply that any graph of genus $g$ admits a proper vertex coloring with at most $H(g)$ colors.


\begin{lemma}
\label{lem:c1-bound}
    If $| \mathcal{C}_1 | \geq H(g)\cdot k$, then $(G, I, J)$ is a \YES-instance.
\end{lemma}
\begin{proof}
    Assume that $| \mathcal{C}_1 | \geq H(g) \cdot k$. Since $G$ admits a proper vertex coloring with $H(g)$ colors, $G[\mathcal{C}_1]$ contains an independent set $I_m$ of size $k$. We may greedily move all tokens of $I$ to $I_m$ and do the same for the tokens of $J$, starting with the tokens adjacent to $I_m$.
\end{proof}

From now on we assume that $|\mathcal{C}_1| \leq H(g)\cdot k$. We bound the size of $\mathcal{C}_3$ using Euler's formula.

\begin{lemma}
\label{lem:c3-bound}
    $|\mathcal{C}_3| \leq 16g^2 + 8g(2k-1) + 8k$.
\end{lemma}

\begin{proof}
    We will bound the number of sets $Y$ such that $Y \subseteq X, |Y| \geq 3$ and $\mathcal{C}_Y \neq \emptyset$.
    For this purpose we 
    consider a drawing of $G$ on a surface of genus $g$.
    We construct a graph $G'$ from $G$ as follows. First, take the subgraph of $G$ induced by the vertices $X$ and for each $Y \subseteq X$ such that $|Y| \geq 3$ and $\mathcal{C}_Y \neq \emptyset$ the vertex $v_Y$ whose neighborhood is $Y$. For each such $Y$, the embedding induces a cyclic ordering $\{v_1, \cdots, v_t\}$ of the neighbors of $v_Y$. We create an edge from $v_i$ to $v_{i+1}$ following the path $v_i,v_Y, v_{i+1}$ for all $i \in \{1, \cdots, t-1\}$ and an edge from $v_1$ to $v_t$ following the path $v_1, v_Y, v_t$. We then remove the vertex $v_Y$ and eliminate parallel edges, creating in the process a unique face $f_Y$ that we associate to the set $Y$. The resulting graph is $G'$.

    Let $G''$ be a triangulation of $G'$. Then $2|E(G'')| = \sum_{v \in V(G'')} d(v) = 3 |F(G'')|$ and thus $|E(G'')| = 3|F(G'')|/2$.
    Let $g'$ be the genus of $G''$ and observe that $g' \leq g$.
    By Euler's formula, 
    $2-2g \le 2-2g'=|V(G'')|-|E(G'')|+|F(G'')|=|V(G'')|-|F(G'')|/2$.  Thus, $|F(G')|\le |F(G'')|\le 2|V(G'')|+4(g-1)\le 2|X|+4(g-1)$.

    Hence, there are at most $2|X| + 4(g-1)$ sets $Y \subseteq X$ such that $|Y| \geq 3$ and $C_Y \neq \emptyset$. By \cite{bouchet1978orientable,ringel1965dasgd}, if a graph has genus $g$, it contains no $K_{3,m}$ as a subgraph when $m \geq 4g+3$. 
    So we obtain
    \begin{align*}
        |\displaystyle \bigcup_{\substack{Y \subseteq X\\|Y| \geq 3}} C_{Y}| &\leq (4g + 2)(2|X| + 4(g-1))\\
        &= 16g(g-1) + 8g|X| + 8g + 4|X| - 8\\
        &\leq 16g^2 - 8g + 16gk + 8k\\
        &= 16g^2 + 8g(2k - 1) + 8k.
    \end{align*}
    \aftermath
\end{proof}

Let $\mathcal{P}$ be the set of pairs $u,v$ such that $u,v \in X$ and $C_{\{u,v\}} \neq \emptyset$.

\begin{lemma}
\label{lem:c2-bound}
    $|\mathcal{P}| \leq 3|X| + 6(g-1)$.
\end{lemma}

\begin{proof}
    Similarly to the proof of \cref{lem:c3-bound}, we draw a graph $G'$ such that $G'$ is obtained by restricting $G$ to the vertices of $X$ such that they are in a pair $Y \in \mathcal{P}$ and one vertex $v_Y$ in $C_Y$ for each $Y \in \mathcal{P}$. For each $\{u,v\} \in \mathcal{P}$, we add an edge $u,v$ following the path $u,v_Y,v$ and then remove the vertex $v_Y$. We then triangulate $G'$ to get $G''$ and apply Euler's formula. The number of edges in $G'$ correspond to the number of elements
    in $\mathcal{P}$; that is, $|\mathcal{P}|=|E(G'')|$.  Thus, the lemma follows from the inequality below.  Since $g(G'')\le g$ and $3|F(G'')|=2|E(G'')|$, we have

    \begin{align*}
        |E(G'')| &\le |V(G'')| + |F(G'')| - (2- 2g)\\
        &= 3|V(G'')|-3(2-2g)\\
        &\leq 3|X| + 6(g-1).
    \end{align*}
    \aftermath
\end{proof}

To deal with $\mathcal{C}_2$, for each $Y \in \mathcal{P}$ such that $|C_Y|$ is large enough, we will use \cref{alg:filtering-bis} to find a vertex set $T_Y$ in $C_Y$ of size $O(k)$ such that we can replace $C_Y$ by $T_Y$ in order to create an equivalent instance $(G', I, J)$, where $G'$ is a subgraph of $G$.

Consider an arbitrary 2-cell embedding of $G$ on an orientable surface of genus~$g$ and a corresponding crossing-free drawing $D$ 
of $G$.  For a pair $Y = \{u,v\} \in \mathcal{P}$, consider the restriction $D[C_Y]$ of $D$ to the drawing of the subgraph $G[C_Y]$. The drawing $D[C_Y]$ partitions the paths $u - w - v$ into homotopy classes. Let $\mathcal{C}$ be any such homotopy class and let 
$C$ be the set of midpoints of the paths in $\mathcal{C}$. 
Since $D$ corresponds to a 2-cell embedding, the restriction of $D$ to the paths in $\mathcal{C}$ is homeomorphic to a planar drawing of $\mathcal{C}$. Therefore, there exists an outer face $f_0$ in the restriction of $D$ to $\mathcal{C}$. 
If more than $2$ vertices in $C$ touch $f_0$, say $w$, $w'$, and $w''$, then we obtain an embedding of $K_{3,3}$ in the plane by placing a vertex inside $f_0$ and connecting it to $w$, $w'$, and $w''$, which is impossible. Therefore, at most two vertices of $C$ touch $f_0$. We call those vertices \emph{outer} and all other vertices of $C$ \emph{inner}. We denote by $N^p(\{u,v\})$ the set of inner vertices corresponding to the pair $u, v$. By iteratively removing outer vertices and their corresponding paths in $\mathcal{C}$ from the drawing (breaking ties arbitrarily), we obtain a linear order $w_1, w_2, \ldots, w_\ell$ on the midpoints~$C$. 

Let $u, v$ be an arbitrary pair in $\mathcal{P}$. 
By the discussion above, the homotopy classes of the paths $u - w - v$ partition the drawing $D$ restricted to the paths $u - w - v$ into (not necessarily connected) regions of the surface that we call \emph{zones}. 
By the following lemma, the number of homotopy classes, and therefore the number of zones, is bounded by a function of $g$ (see \cite[Proposition 4.2.7]{mohar2001graphs} and the discussion after its proof).
Let $f_M(g) := \max \{1, 4g \}$.

\begin{lemma}[{see \cite[Proposition 4.2.7]{mohar2001graphs}}]
\label{thm:mohar-homotopy}
Let $G$ be a graph embedded on an orientable surface of genus $g$. If $P_1, \cdots, P_k$ 
    are $u$-$v$ paths of $G$ ($u \neq v$) such that no two are homotopy-equivalent with respect to the embedding, then $k \leq f_M(G)$.
\end{lemma}

\Cref{fig:zones-torus} shows a drawing of a set of $u$-$v$ paths on the torus having $8$ outer vertices and $4$ zones.
Clearly, there are no edges between vertices in two different zones of the same pair $u, v$ in $\mathcal{P}$.
The following lemma confirms the intuition that any two inner vertices of two zones that correspond to different pairs in $\mathcal{P}$ are non-adjacent.

\begin{figure}
    \centering
    \begin{tikzpicture}[scale=0.6]
	\begin{pgfonlayer}{nodelayer}
		\node [style=none] (0) at (-5, 5) {};
		\node [style=none] (1) at (5, 5) {};
		\node [style=none] (2) at (5, -5) {};
		\node [style=none] (3) at (-5, -5) {};
		\node [style=base] (4) at (0, 2) {$u$};
		\node [style=base] (5) at (0, -2) {$v$};
		\node [style=base] (6) at (-1.5, 0) {};
		\node [style=base] (7) at (1.5, 0) {};
		\node [style=base] (10) at (-1.25, 4) {};
		\node [style=base] (11) at (1.25, 4) {};
		\node [style=none] (12) at (1, 5) {};
		\node [style=none] (13) at (-1, 5) {};
		\node [style=none] (14) at (-1, -5) {};
		\node [style=none] (15) at (1, -5) {};
		\node [style=none] (16) at (5, 1) {};
		\node [style=none] (17) at (5, -1) {};
		\node [style=base] (18) at (4, 1.75) {};
		\node [style=base] (19) at (4, -1) {};
		\node [style=none] (20) at (-5, 1) {};
		\node [style=none] (21) at (-5, -1) {};
		\node [style=base] (22) at (3.25, -2) {};
		\node [style=base] (23) at (2.5, -4) {};
		\node [style=none] (24) at (5, -3) {};
		\node [style=none] (25) at (5, -4.25) {};
		\node [style=none] (26) at (-5, -4.25) {};
		\node [style=none] (27) at (-5, -3) {};
		\node [style=none] (28) at (-4.25, -5) {};
		\node [style=none] (29) at (-3, -5) {};
		\node [style=none] (30) at (-3, 5) {};
		\node [style=none] (31) at (-4.25, 5) {};
	\end{pgfonlayer}
	\begin{pgfonlayer}{edgelayer}
		\draw [style=arrow, line width=1pt, dashed] (0.center) to (1.center);
		\draw [style=arrow, line width=1pt, dashed] (2.center) to (1.center);
		\draw [style=arrow, line width=1pt, dashed] (3.center) to (2.center);
		\draw [style=arrow, line width=1pt, dashed] (3.center) to (0.center);

        \draw [opacity=0.5, fill=coolred, line width=2pt, draw=coolred] (5.center) -- (6.center) -- (4.center) -- (7.center) -- (5.center);
        
        \draw [opacity=0.5, fill=coolblue, line width=2pt, draw=coolblue] (4) -- (10.center) -- (13.center) -- (12.center) -- (11.center) -- (4);
        \draw [opacity=0.5, fill=coolblue, line width=2pt, draw=coolblue] (5) -- (14.center) -- (15.center) -- (5);

        \draw [opacity=0.5, fill=coolorange, line width=2pt, draw=coolorange] (4) -- (18.center) -- (16.center) -- (17.center) -- (19.center) -- (4);
        \draw [opacity=0.5, fill=coolorange, line width=2pt, draw=coolorange] (5) -- (20.center) -- (21.center) -- (5);

        \draw [opacity=0.5, fill=coolorange, line width=2pt, draw=coolorange] (4) -- (18.center) -- (16.center) -- (17.center) -- (19.center) -- (4);
        \draw [opacity=0.5, fill=coolorange, line width=2pt, draw=coolorange] (5) -- (20.center) -- (21.center) -- (5);

        \draw [opacity=0.5, fill=coolpink, line width=2pt, draw=coolpink] (4) -- (30.center) -- (31.center) -- (4);
        \draw [opacity=0.5, fill=coolpink, line width=2pt, draw=coolpink] (27.center) -- (26.center) -- (28.center) -- (29.center) -- (27.center);
        \draw [opacity=0.5, fill=coolpink, line width=2pt, draw=coolpink] (5) -- (22.center) -- (24.center) -- (25.center) -- (23.center) -- (5);
        
		\draw (5) to (6);
		\draw (6) to (4);
		\draw (4) to (7);
		\draw (7) to (5);
		\draw (4) to (10);
		\draw (4) to (11);
		\draw (15.center) to (5);
		\draw (10) to (13.center);
		\draw (11) to (12.center);
		\draw (5) to (14.center);
		\draw (4) to (19);
		\draw (4) to (18);
		\draw (18) to (16.center);
		\draw (17.center) to (19);
		\draw (20.center) to (5);
		\draw (21.center) to (5);
		\draw (5) to (22);
		\draw (5) to (23);
		\draw (22) to (24.center);
		\draw (23) to (25.center);
		\draw (26.center) to (28.center);
		\draw (27.center) to (29.center);
		\draw (30.center) to (4);
		\draw (31.center) to (4);

        \node [style=none] (y1) at (barycentric cs:5=0.5,6=0.5,4=0.5,7=0.5) {$Y_1$};
        
        \node [style=none] (y21) at (barycentric cs:4=0.5,10=0.5,13=0.5,12=0.5,11=0.5) {$Y_2$};
        \node [style=none] (y22) at (barycentric cs:5=0.5,14=0.5,15=0.5) {$Y_2$};

        \node [style=none] (y31) at (barycentric cs:4=0.5,18=0.5,16=0.5,17=0.5,19=0.5) {$Y_3$};
        \node [style=none] (y32) at (barycentric cs:5=0.5,20=0.5,21=0.5) {$Y_3$};

        \node [style=none] (y41) at (barycentric cs:4=0.5,30=0.5,31=0.5) {$Y_4$};
        \node [style=none] (y42) at (barycentric cs:27=0.5,26=0.5,28=0.5,29=0.5) {$Y_4$};
        \node [style=none] (y43) at (barycentric cs:5=0.5,22=0.5,24=0.5,25=0.5,23=0.5) {$Y_4$};
        
	\end{pgfonlayer}
\end{tikzpicture}
    \caption{Four zones of a pair $Y = \{u,v\}$ on a torus ($g = 1$).}
    \label{fig:zones-torus}
\end{figure}
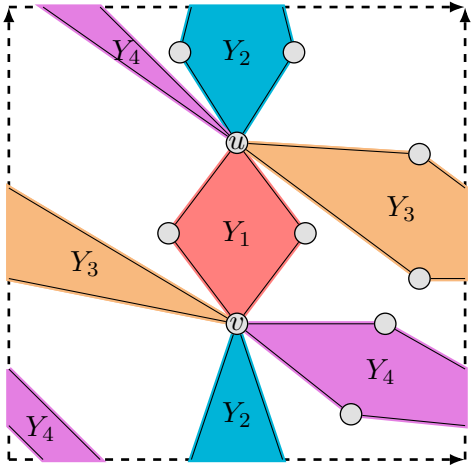

\begin{lemma}
    \label{lem:inners-is-linear-forest}
    If $Y \in \mathcal{P}$, then $G[N^p(Y)]$ is a linear forest.
\end{lemma}

\begin{proof}
    If $w_i$ is an inner vertex of $Y$ in a zone $A$, then $w_i$ only has neighbors in the zone $A$ (since $w_i$ lies inside of a $4$-cycle induced by the two outer vertices of $A$ and the two vertices of $Y$, and all vertices inside of this $4$-cycle are inner vertices of $A$, because by definition all vertices of $A$ are in the same homotopy class).
    Furthermore, $w_i$ can only be adjacent in $G[C_Y]$ to the previous and next vertices of that zone: $w_{i-1}$ and $w_{i+1}$, if they exist. Thus, $G[N^p(Y)]$ has maximum degree at most $2$ and contains no cycles.
\end{proof}

We will now reduce the size of $\mathcal{C}_2$ to $O(g^2 + gk + k^2)$. Fix $Y \in \mathcal{P}$ such that $|C_Y| \geq 2f_M(g) + 4k + 4$. Notice the number of inner vertices of $Y$ is at least $4k + 4$, because the total number of outer vertices of $Y$ (over all zones) is at most $2f_M(g)$. Note that a vertex $v \not\in X$ is either inside a zone of $Y$, or else outside all its zones. 
If $v$ is outside, then $v$ is adjacent to none of the inner vertices of $C_Y$.  But if $v$ is inside a zone of $Y$, then $v$ is adjacent to at most two vertices of $C_Y$.

First, we will find in polynomial time an independent subset $T_Y$ of $C_Y$ of size $2k + 2$, such that each vertex $v \not\in (C_Y \cup Y)$ is adjacent to at most two elements of $T_Y$. Then, we will show that we can restrict $C_Y$ to $T_Y$ and produce an equivalent instance. In our proof, $T_Y$ plays a similar role to black holes introduced in Bartier et al.~\cite{bartier:hal-04323590} for the Token Sliding problem. The vertex set $T_Y$ will absorb any tokens moved to $C_Y$ and all
vertices adjacent to $C_Y$ will be free to receive tokens.

Let $G'$ be the graph formed as follows: for each $Y \in \mathcal{P}$ such that $|C_Y| \geq 2f_M(g) + 4k + 4$, we remove the vertices of $C_Y \setminus T_Y$.

We first show we can compute $T_Y$ in polynomial time.

\begin{algorithm}[H]    
\caption{Filtering algorithm}\label{alg:filtering-bis}
\DontPrintSemicolon
\KwIn{Instance $(G, I, J)$ of \textsc{Token Jumping}, a pair $Y \in \mathcal{P}$ with $|C_Y| \geq 2f_M(g)+4k+4$.}
\KwOut{A linear forest $Z_Y$ such that $|Z_Y| \geq N^p(Y)$ and any vertex $v \not\in (C_Y \cup Y)$ has at most 2 neighbors in $Z_Y$.}
\SetKw{KwGoTo}{go to line}
$Z := C_Y$\\ 
\For{$v \in V(G) - (C_Y \cup Y)$}{
    \lIf{$v$ has at least $3$ neighbors in $C_Y$}{$Z \leftarrow Z - N(v)$}\label{line:alg-ext-degree-cond}
}
\For{$w \in Z \cap C_Y$}{
    \lIf{$w$ has degree at least $3$ in $G[Z]$}{$Z \leftarrow Z' - w$}\label{line:alg-degree-cond}
}
    Remove arbitrarily one vertex\label{line:cycle-cond}\footnote{It is natural to want to remove here an \emph{outer} vertex.  However, the algorithm does not have as part of its input an embedding of the graph, so it may not be possible to quickly determine 
    for each zone which vertices are its outer vertices.} 
    from each cycle in $G[Z]$\\
\Return $Z$
\end{algorithm}

\begin{lemma}
\label{lem:alg-gives-linforest}
    Let $Z_Y$ be the output of \cref{alg:filtering-bis} for a pair $Y \in \mathcal{P}$. Then $G[Z_Y]$ is a linear forest such that any vertex $v \not\in (C_Y \cup Y)$ has at most $2$ neighbors in $Z_Y$.
\end{lemma}

\begin{proof}
    First, we show $G[Z_Y]$ is a linear forest. Assume for a contradiction that, just before reaching~\cref{line:cycle-cond}, there exists a vertex $w \in Z_Y$ such that $d_{G[Z_Y]}(w) \geq 3$, then $w$ would have been removed by \cref{line:alg-degree-cond}. It follows that when reaching~\cref{line:cycle-cond}, the graph $G[Z_Y]$ is of maximum degree 2. The algorithm then removes one vertex from each cycle of $G[Z_Y]$, so the vertex set $Z_Y$ returned by~\cref{alg:filtering-bis} induces a linear forest.

    Finally, \cref{line:alg-ext-degree-cond} ensures that every vertex $v \not\in (C_Y \cup Y)$ has at most $2$ neighbors in $Z_Y$.
\end{proof}

\begin{lemma}
\label{lem:zy-has-correct-size}
    Let $Z_Y$ be the output of \cref{alg:filtering-bis} for a pair $Y \in \mathcal{P}$. We have $|Z_Y| \geq |N^p(Y)|$.
\end{lemma}

\begin{proof}
    We will first show \cref{line:alg-ext-degree-cond} only removes outer vertices of $C_Y$. Let $v$ be a vertex such that $v \not\in (C_Y \cup Y)$. If $v$ is inside some zone of $C_Y$, then $v$ can have at most two neighbors in $C_Y$ (these are the predecessor and successor of $v$ in the linear order constructed after \cref{lem:c2-bound}). If $v$ is outside of any zone of $C_Y$, then $v$ cannot be adjacent to any inner vertex of $C_Y$ (because the the two outer vertices of $C_Y$, together with the two
    vertices in $Y$, induce a cycle that separates $v$ from these inner vertices). So \cref{line:alg-ext-degree-cond} can only remove outer vertices of $C_Y$. Because an inner vertex can only be adjacent to two vertices of $C_Y$, \cref{line:alg-degree-cond} can only remove outer vertices. To conclude, it remains to show that we can associate to each inner vertex removed in \cref{line:cycle-cond} a unique outer vertex. If an inner vertex $w$ of a zone of $Y$ is removed, then this was to
    break a cycle that contained all vertices in the zone; so all these vertices were in $Z_Y$. Therefore, $Z_Y$ contains the two outer vertices of that zone and they appear in only one cycle. We can therefore associate to $w$ either of these two outer vertices of the zone.
\end{proof}

By \cref{lem:alg-gives-linforest,alg:filtering-bis} outputs a vertex set $Z_Y$ of size at least $4k+4$ such that $G[Z_Y]$ is a linear forest, thus we can find in linear time an independent set $T_Y$ contained in $Z_Y$ of size $2k+2$.

\begin{lemma}
    \label{lem:size}
    $|V(G')| = O(g^2 + gk + k^2)$.
\end{lemma}

\begin{proof}

    The size of $\mathcal{C'}_2$ in $G'$ can be bounded directly as each pair $Y$ either satisfied $|C_Y| \leq 2f_M(g) + 4k + 4$ or was restricted by \cref{alg:filtering-bis} to a subset $T_Y$ of size $2k + 2$.
\begin{align*}
|V(G')| &= |X| + |\mathcal{C}_1| + |\mathcal{C}_3| + |\mathcal{C}'_2|\\
&\leq 2k + H(g)k + (16g^2 + 8g(2k-1) + 8k) + (|\mathcal{P}|(2f_M(g)+4k+4))\\
&\leq 16g^2 + 8g(2k-1) + H(g)k + 10k + ((3|X|+6(g-1))(2f_M(g)+4k+4))\\
&\leq 16g^2 + 8g(2k-1) + H(g)k + 10k + 12((k+(g-1))(f_M(g)+2k+2))\\
&= 16g^2 + 16gk\! -\! 8g\! +\! H(g)k\! +\! 10k\! +\! 12(f_M(g)k\! +\! 2k^2\! +\! f_M(g)g\! +\! 2gk\! +\! 2g\! -\! f_M(g)\! -\! 2)\\
&= 16g^2 + 12f_M(g)g + 40gk + 12f_M(g)k + 16g - 12f_M(g) + H(g)k + 10k + 24k^2 - 24
\end{align*}

\noindent If $g = 0$, then $|V(G')| \leq 24k^2 + 26k$.\\
\noindent If $g \geq 1$, then $|V(G')| \leq 64g^2 + 88gk + H(g)k - 32g + 10k + 24k^2$.

\aftermath
\end{proof}

\medskip

\begin{theorem}
    \label{thm:equivalence}
    The instances $(G, I, J)$ and $(G', I, J)$ of \TJ are equivalent.
\end{theorem}

\begin{proof}
    Since $G'$ is an induced subgraph of $G$, if $(G', I, J)$ is \YES, then also $(G, I, J)$ is \YES.
    It remains to show that if $(G, I, J)$ is \YES, then also $(G', I, J)$ is \YES.

    Assume that $(G, I, J)$ is \YES and let $\sigma$ be a sequence of token jumps certifying this. 
    We will use $\sigma$ to construct a jump sequence $\sigma'$ in $G'$ such that after each step in $\sigma'$ the positions of the tokens agree exactly with those after a corresponding step in $\sigma$, except that tokens in $G'$ (moved by $\sigma'$) may be arbitrarily rearranged among the vertices of $C_Y$ for each $Y$ such that $C_Y\subseteq \mathcal{C}_2$.  We proceed by induction on the length $\ell$ of $\sigma$.  We remark that we may have multiple
    steps in $\sigma'$ that correspond to the same step in $\sigma$.

    For each step in $\sigma$, we try to copy it in $\sigma'$.  If the target vertex $v$ appears in $G'$ and the vertices with tokens (after moving a token to $v$) induce an independent set, then we make the same move to extend $\sigma'$.
There are two reasons why this might not be the case, and we consider them both below.

    (a) If the target vertex $v$ is absent from $G'$, then $v\in C_Y\setminus T_Y$ for some pair $Y$.  So it suffices to show there exists $v'\in T_Y$ such that ($v'\in V(G')$ and) no token currently appears in $N_{G'}[v']$.  To see this, using \cref{lem:alg-gives-linforest}, note that each token appears in $N_{G'}[w]$ for at most two vertices $w\in T_Y$.  Since $|T_Y|=2k+2$, the number of available vertices $v'\in T_Y$ is at least $|T_Y|-2|I| = (2k+2)-2k=2$.
    

    (b) Suppose that $\sigma$ moves a token from a vertex $u$ to a vertex $v$ and this causes a conflict. That is, the token on $u$ moves to a vertex $v$ that is adjacent to a vertex in some $T_Y$ that contains a token due to (a). There may be many $Y$ that are concerned, which will all be handled sequentially. This step would be conceptually simpler (that is, we would only need to 
    consider at most one such $Y$) if we could guarantee that $T_Y$ includes neither outer vertices nor any of their neighbors. 
    However, it is difficult to ensure this, since we have no embedding of the graph.

    Assume that the target vertex $v$ is adjacent to at least one vertex $t_1$ of $T_Y$ with a token (possibly $u = t_1$) and possibly to some other vertex $t_2 \in T_Y$, if it exists. At most $2k$ vertices of $T_Y$ contain a token or have a token in their neighborhood. Therefore, we can move the token on $t_1$ to at least two vertices $t'_1, t'_2 \in T_Y$. However we do not want to move the token on $t_1$ to a neighbor of $v$ as that would not resolve the conflict. But by construction, $v$ is adjacent to $t_1$ and at most one other vertex of $T_Y$. So at least
    one of $t'_1$ and $t'_2$ is not adjacent to $v$, and we move the token from $t_1$ to that vertex. If $v$ is adjacent to another token on $T_Y$, then both $t'_1$ and $t'_2$ are not adjacent to $v$, so we move the tokens on $t_1$ and $t_2$ to $t'_1$ and $t'_2$.
%
\end{proof}

By combining \cref{lem:size} and \cref{thm:equivalence} we obtain \cref{thm:main-result-kernel}.

\section{A kernel for \TJ on $K_{2,3}$-free graphs}

In this section we show that the same techniques as in \cref{sec:kernel} can be used to obtain a sub-quadratic kernel for \TJ on $K_{2,3}$-free graphs, which notably include outerplanar graphs. Notice that \TJ remains \PSPACE-complete on $K_{2,3}$-free graphs by~\cite[Theorem 2]{BBM:24}.
The analysis will be slightly different for $\mathcal{C}_3$ and significantly simpler for $\mathcal{C}_2$. We will use the same notation as in \cref{sec:kernel}.


\begin{theorem}
\label{thm:k23free-main-kernel}
    \TJ parametrized by the size $k$ of the independent sets and the genus $g$ of the $K_{2,3}$-free input graph admits a kernel of size $O(g + \sqrt{g}k + k)$.
\end{theorem}

\begin{proof}
    From \cref{lem:c1-bound}, we know that if $|\mathcal{C}_1| \geq H(g) \cdot k$, we have a \textsc{Yes}-instance. We therefore assume that $|\mathcal{C}_1| < H(g) \cdot k$.
    
    We now show that $|\mathcal{C}_2| \leq 12k+12g$. By \cref{lem:c2-bound}, $|\mathcal{P}| \leq 6k + 6g$. However there are no $K_{2,3}$ subgraphs, so for any $Y \in \mathcal{P}$, we have $|C_Y| \leq 2$. Hence $|\mathcal{C}_2| \leq 12k + 12g$.

    Finally, we show that $|\mathcal{C}_3| \leq 4k + 4g$. The beginning of the proof for \cref{lem:c3-bound} still holds: there are at most $2|X| + 4(g-1)$ sets $Y \subseteq X$ such that $|Y| \geq 3$ and $C_Y \neq \emptyset$. However, since there is no $K_{2,3}$ subgraph, we have that $|C_Y| = 1$ for any such $Y$.

    Combining the bounds above we obtain
    \begin{align*}
        |V(G)| &= |X| + |\mathcal{C}_1| + |\mathcal{C}_2| + |\mathcal{C}_3|\\
        &\leq 2k + H(g)k + (12k + 12g) + (4k + 4g)\\
        &\leq 18g + H(g)k + 18k.
    \end{align*}
\aftermath    
\end{proof}

Observe that all that is required for the kernelization algorithm is to compute the size of $|\mathcal{C}_1|$, which can be done in time $O(k \cdot |E(G)|)$.
Noticing that any outerplanar graph is $3$-colorable, we may assume that $|\mathcal{C}_1| < 3k$, which gives a linear kernel for \TJ on outerplanar graphs (\TJ on outerplanar graphs is not known to be \PSPACE-hard). 
\begin{corollary}
    \TJ parameterized by the size $k$ of the independent sets admits a kernel of size $21k$ on outerplanar graphs.
\end{corollary}

\bibliographystyle{plain}
\bibliography{biblio}

\end{document}